\documentclass[fleqn,11pt,twoside]{article}

\usepackage{amsthm,amsthm,amssymb, color, xcolor,epsfig, graphics, subfigure}

\usepackage{amsmath, graphicx, latexsym, lscape }

\makeatletter
\newcommand{\copyrightnote}[2]{{\renewcommand{\thefootnote}{}
 \footnotetext{\small\it
\begin{flushleft}
 \copyright \ #1   #2
\end{flushleft}}}}

\newcommand{\Name}[1]{\begin{flushleft}
                       \LARGE \bf #1
                       \end{flushleft}\vspace{-3mm}}

\newcommand{\Author}[1]{\begin{flushleft}
                       \it #1 \end{flushleft}}

\newcommand{\Address}[1]{\begin{flushleft}
                       \it #1 \end{flushleft}}

\newcommand{\Date}[1]{\begin{flushleft}
                      \small  \it #1 \end{flushleft}}

\newcommand{\evenhead}{Author \ name}
\newcommand{\oddhead}{Article \ name}

\renewcommand{\@evenhead}{
\hspace*{-3pt}\raisebox{-15pt}[\headheight][0pt]{\vbox{\hbox to \textwidth
{\thepage \hfil \evenhead}\vskip4pt \hrule}}}
\renewcommand{\@oddhead}{
\hspace*{-3pt}\raisebox{-15pt}[\headheight][0pt]{\vbox{\hbox to \textwidth
{\oddhead \hfil \thepage}\vskip4pt\hrule}}}
\renewcommand{\@evenfoot}{}
\renewcommand{\@oddfoot}{}

\long\def\@makecaption#1#2{%
  \vskip\abovecaptionskip
  \sbox\@tempboxa{\small \textbf{#1.}\ \ #2}%
  \ifdim \wd\@tempboxa >\hsize
    {\small \textbf{#1.}\ \ #2}\par
  \else
    \global \@minipagefalse
    \hb@xt@\hsize{\hfil\box\@tempboxa\hfil}%
  \fi
  \vskip\belowcaptionskip}

\newcommand{\JNMPnumberwithin}[3][\arabic]{%
  \@ifundefined{c@#2}{\@nocounterr{#2}}{%
    \@ifundefined{c@#3}{\@nocnterr{#3}}{%
      \@addtoreset{#2}{#3}%
      \@xp\xdef\csname the#2\endcsname{%
        \@xp\@nx\csname the#3\endcsname .\@nx#1{#2}}}}%
}

\renewenvironment{proof}[1][\proofname]{\par
  \normalfont
  \topsep6\p@\@plus6\p@ \trivlist
  \item[\hskip\labelsep\textbf{%
    #1\@addpunct{.}}]\ignorespaces
}{%
  \qed\endtrivlist
}

\newcommand{\resetfootnoterule} {
  \renewcommand\footnoterule{%
  \kern-3\p@
  \hrule\@width.4\columnwidth
  \kern2.6\p@}
}

\renewcommand{\footnoterule}{}

\makeatother

\theoremstyle{definition}

\newtheorem{rmk}{Remark}
\newtheorem{thm}{Theorem}
\newtheorem{coro}{Corollary}

\begin{document}

\renewcommand{\evenhead}{ {\LARGE\textcolor{blue!10!black!40!green}{{\sf \ \ \ ]ocnmp[}}}\strut\hfill C Muriel and M C  Nucci}
\renewcommand{\oddhead}{ {\LARGE\textcolor{blue!10!black!40!green}{{\sf ]ocnmp[}}}\ \ \ \ \   Generalized Symmetries, First Integrals, and  Exact Solutions}

\thispagestyle{empty}
\newcommand{\FistPageHead}[3]{
\begin{flushleft}
\raisebox{8mm}[0pt][0pt]
{\footnotesize \sf
\parbox{150mm}{{Open Communications in Nonlinear Mathematical Physics}\ \  \ \ {\LARGE\textcolor{blue!10!black!40!green}{]ocnmp[}}
\quad Vol.1 (2021) pp
#2\hfill {\sc #3}}}\vspace{-13mm}
\end{flushleft}}

\FistPageHead{1}{\pageref{firstpage}--\pageref{lastpage}}{ \ \ Article}

\strut\hfill

\strut\hfill

\copyrightnote{The author(s). Distributed under a Creative Commons Attribution 4.0 International License}

\Name{Generalized Symmetries, First Integrals, and  Exact Solutions of Chains of Differential Equations}

\Author{C Muriel$^{\,1}$ and M~C Nucci$^{\,2}$}

\Address{$^{1}$ Departamento de Matem\'aticas, Universidad de C\'adiz, 11510 Puerto Real, Spain\\[2mm]
$^{2}$ Dept. Math. Inf. Phys. Earth Sci. (MIFT), University of Messina, 98166 Messina, \\\& INFN, Section of Perugia, 06123 Perugia, Italy}

\Date{Received 14 April 2021; Accepted 12 June 2021}

\setcounter{equation}{0}

\begin{abstract}
\noindent
New integrability properties of a family of sequences of ordinary differential equations, which contains the Riccati and Abel chains as the most simple sequences, are studied. The determination of $n$ generalized symmetries of the $n$th-order equation in each chain provides, without any kind of integration, $n-1$ functionally independent first integrals of the equation. A remaining first integral arises by a quadrature by using a Jacobi last multiplier that is expressed in terms of the preceding equation in the corresponding sequence.  The complete set of $n$ first integrals is used to obtain the exact general solution of the $n$th-order equation of each sequence.
The results are applied to derive directly the exact general solution of any equation in the Riccati and Abel chains.
\end{abstract}

\label{firstpage}


\section{Introduction}\label{seccion2}
For a given  smooth function $g=g(u)$ defined on some open interval $J\subset \mathbb{R},$ let us  define the differential operator
\begin{equation}\label{operador}
\mathbb{D}_g=\mathbf{D}_t+g(u),
\end{equation}  where $\mathbf{D}_t$ denotes the total derivative operator
\begin{equation}
\label{dt}
\mathbf{D}_t=\partial_t+u_1\partial_u+\cdots+u_k\partial_{u_{k-1}}+\cdots,
\end{equation}
and $u_i=\dfrac{d^{i}u}{d t^i},$ for $i\in\mathbb{N}.$

The differential operator $\mathbb{D}_g$ acts on the set of smooth functions $u=u(t),$  defined in some open interval $I\subset \mathbb{R}$  such that $u(t)\in J$ for $t\in I.$ We set $\mathbb{D}_g^0u=u$ and, for $j\ge 1$, $\mathbb{D}_g^j(u)=\mathbb{D}_g(\mathbb{D}_g^{j-1}(u))$.

We consider the sequence of ordinary differential equations (ODEs)
\begin{equation}\label{cadenageneral}\mathbb{D}^n_g u=0, \qquad (n\in \mathbb{N})\end{equation}
obtained by applying successively the differential operators in the sequence $\{\mathbb{D}_g^n\}_{n\in \mathbb{N}}$ to an unknown function $u=u(t).$
In what follows the sequence of ODEs  $\mathbb{E}_g:=\{\mathbb{D}_g^n(u)=0\}_{n\in \mathbb{N}}$ will be called  \textsl{the chain generated by the function $g.$}

Recursion operators firstly appeared in the context of evolution equations in two independent variables \cite{olver1977} and are traditionally applied to partial differential equations. M Euler et al.   \cite{euler2007riccati} adapted them to ordinary differential equations, by considering a 1+1 evolution equation and a known recursion operator free of $t.$ One the most well known sequences of the   form (\ref{cadenageneral}) is the chain generated by the function $g(u)=ku,$ where $k \in\mathbb{R},$ which is known in the literature as the Riccati chain  of parameter $k\in \mathbb{R}$ and whose elements are usually called   the higher Riccati equations \cite{andriopoulosdifferential,carinena2009geometric,euler2007riccati,leach2009novel}. Riccati sequence arises from a known recursion operator of a second-order evolution equation (Class VIII in \cite{EulerMyN}) linearizable by a $x$-generalized hodograph transformation \cite{euler2003linearizable}. The Riccati sequence exhibits interesting properties from both physical and mathematical point of view, and a large variety of methods (Darboux factors, Jacobi multipliers, extended Prelle-Singer methods, nonlocal symmetries, etc.) have been applied to its study   \cite{ andriopoulos2008generalised,andriopoulosdifferential, carinena2009geometric}.

For different functions $g=g(u)$ the differential operator (\ref{operador}) generates new sequences $\mathbb{E}_g$, some of whose properties have been obtained in   \cite{muriel2014lambda}. In particular, it was proved the existence of a common $\mathcal{C}^{\infty}-$symmetry \cite{muriel01ima1} for  all the equations of the sequence $\mathbb{E}_g:$ the pair $(\mathbf{v}, {\lambda}),$ where
\begin{equation}
\label{lambdasym}\mathbf{v}=\partial_u\quad \mbox{and}\quad  {\lambda}=\dfrac{u_1}{u}-u\,g'(u)
\end{equation} defines a $\mathcal{C}^{\infty}-$symmetry   of the $n$th-order equation  $\mathbb{D}^n_g u=0,$ for $n\in \mathbb{N}$   \cite[Theorem 2.1]{muriel2014lambda}.  This $\mathcal{C}^{\infty}-$symmetry was used to connect  any sequence $\mathbb{E}_g$ with the Riccati chain  of parameter $k\in \mathbb{R}.$


In this paper we aim to investigate  some integrability properties of the class of chains generated by functions of the form $g(u)=ku^m,$  where $k \in\mathbb{R}$ and $m\in\mathbb{Z}.$ The corresponding elements of this type of sequences  will be denoted by  $P_n=0,$ where 
\begin{equation}\label{cadenaparticular}P_0:=u,\quad P_n:=(\mathbf{D}_t+ku^m)^n(u), \qquad n=1,2,\dots.\end{equation}
It should be observed that the Riccati chain is contained in the study, as well as the  chain  generated by the function  $g(u)=ku^2$, $u\in\mathbb{R}$, which is known as the Abel chain  of parameter $k\in \mathbb{R}$  and that has also been extensively studied in the literature \cite{Carinena2011,Carinena2009Abel}. Both chains contain as particular cases well-known families of equations in Mathematical Physics, such as Emden equations, generalized Van-der Pol oscillators, Chazy equations, etc.

For the purposes of this paper, a notable property of the $n$th-order equation $P_n=0$ is the existence of a Jacobi last multiplier \cite{nucci2005jacobi} that  can be expressed in terms of the previous element of the sequence as follows \cite[Theorem 3.3]{muriel2014lambda}:
\begin{equation}
\label{jlm}
M_n:=(P_{n-1})^{-(n+m)}.
\end{equation}
Both elements, the Jacobi last multiplier \eqref{jlm} and the common $\mathcal{C}^{\infty}-$symmetry (\ref{lambdasym}),
will be exploited in the study performed in this paper, which is organized as follows. In Section \ref{section2},  we firstly use the corresponding common $\mathcal{C}^{\infty}-$symmetry \eqref{lambdasym}  to derive  $n$ generalized symmetries of the $n$th-order equation $P_n=0.$
In Section \ref{section3} it is proved that these generalized symmetries have the property that any ratio of their characteristics is a first integral of the equation $P_n=0;$ furthermore, $n-1$ of such first integrals are functionally independent (Theorem \ref{teofirstintegrals}).  The Jacobi last multiplier (\ref{jlm}) is used in Section \ref{section4} to compute a remaining first integral, which constitutes, together with the $n-1$ previously determined, a complete set of first integrals of $P_n=0.$ As a consequence, we obtain explicitly the general solution of all the equations  in any of the chains of the family.  The results are applied to study the Riccati and Abel chains in sections  \ref{seccionR} and \ref{seccionA}, respectively.  Additionally, we prove that  for each equation in the Riccati chain it is possible to derive an additional generalized symmetry. As a consequence, a complete set of first integral for the $n$th-order equation in the Riccati chain is obtained without any kind of integration.

Remarkably, the unified procedure presented in this paper is valid for all the chains generated by $g(u)=ku^m,$, providing a common expression, depending on $m,$ for the exact solutions of all the equations in any chain. This greatly improves the procedure presented in \cite{muriel2014lambda}, which required an additional integration to derive such solutions from the solutions of the higher Riccati equations. Moreover, our results are valid for real values of the parameter $m$ (Section \ref{section7}), enlarging significantly the classes of ODEs which can be completely solved by this new procedure, without any kind of integration.

\section{Derivation of $n$ generalized symmetries by using a
	\\ $\mathcal{C}^{\infty}-$symmetry}\label{section2}

In this section we address the problem of determining $n$ generalized symmetries of the $n$th-order equation: 
\begin{equation}
\label{Pn} P_n:=(\mathbf{D}_t+ku^m)^n(u)=0,\qquad k \in\mathbb{R}, \quad m\in\mathbb{Z},
\end{equation}
by using the corresponding $\mathcal{C}^{\infty}-$symmetry  (\ref{lambdasym}), which for  $g(u)=ku^m$  becomes
\begin{equation}
\label{lambdasympar}\mathbf{v}=\partial_u\quad \mbox{and}\quad  {\lambda}=\dfrac{u_1}{u}-k\,m\, u^m.
\end{equation}

We recall that the concept of $\mathcal{C}^\infty-$symmetry arises naturally from the concept of Lie point symmetry when $\lambda-$prolongations are considered instead of standard prolongations: for any vector field $\mathbf{v}=\xi(t,u)\partial_t+\eta(t,u)\partial_u$ and a smooth function $\lambda=\lambda(t,u,u_1),$  the $j$th--order $\lambda-$prolongation of $\mathbf{v}$ is denoted by $\mathbf{v}^{[\lambda,(j)]}$ and defined as
\begin{equation}\label{lambdaprol1}
\mathbf{v}^{[\lambda,(j)]}:=\mathbf{v}+\displaystyle \eta^{[\lambda,(1)]}\big(t,u^{(1)}\big)\partial_{u_1}+\cdots+\eta^{[\lambda,(j)]}\big(t,u^{(j)}\big)\partial_{u_j},\end{equation}
where
\begin{equation}\label{lprol2}
\eta^{[\lambda,(0)]}=\eta,\qquad \eta^{[\lambda,(i)]}=(\mathbf{D}_t+\lambda)\big(\eta^{[\lambda,(i-1)]}\big)-u_i(\mathbf{D}_t+\lambda)\big(\xi\big),
\end{equation}
for $i=1,\dots,j$.
The pair $(\mathbf{v},\lambda)$ defines a $\mathcal{C}^\infty-$symmetry (or  $\mathbf{v}$ is a $\lambda-$symmetry) of equation (\ref{Pn})   if and only if \cite{muriel01ima1}
\begin{equation}\label{lambdasim001}
\mathbf{v}^{{[\lambda,(n)]}}\big(P_n\big)=0\quad \mbox{ when }\quad  P_n=0.
\end{equation}
Equivalently,  $(\mathbf{v},\lambda)$ is a $\mathcal{C}^\infty-$symmetry of (\ref{Pn}) if and only if \cite{muriel01ima1}
\begin{equation}\label{lambdasim010}
[\mathbf{v}^{{[\lambda,(n-1)]}},\mathbf{A}_n]=\lambda \mathbf{v}^{{[\lambda,(n-1)]}}-(\mathbf{A}_n+\lambda)(\xi)\mathbf{A}_n,
\end{equation} where now, and henceforth, $\mathbf{A}_n$ denotes  the restriction of $\mathbf{D}_t$ to the manifold defined by equation (\ref{Pn}).

This primitive  concept of  $\mathcal{C}^{\infty}-$symmetry was later extended to permit the function $\lambda$ or the infinitesimals  $\xi,\eta$ of $\mathbf{v}$ belong to the space  of smooth functions on $t, u$ and the derivatives of $u$ with respect to $t$ up to some finite but unspecified order
\cite[Def. 2.1]{muriel03lie},
\cite[Sect. 2.3]{muriel2012SIGMA}.  In these cases the pair $(\mathbf{v},\lambda)$ is called a generalized $\lambda-$symmetry (or a generalized $\mathcal{C}^{\infty}-$symmetry). Clearly, when $\lambda=0,$ \eqref{lambdaprol1} is the standard prolongation of  $\mathbf{v}$ and therefore the pair $(\mathbf{v},0)$ is a (generalized) $\mathcal{C}^\infty-$symmetry if and only if $\mathbf{v}$ is a (generalized) Lie symmetry.

A remarkable property of $\mathcal{C}^{\infty}-$symmetries for the purpose of this section is the following \cite{muriel03lie}:  if $(\mathbf{v},{\lambda})$ is a $\mathcal{C}^{\infty}-$symmetry of equation (\ref{Pn}) and $\rho=\rho(t,u^{(i-1)}),$ $i\leq n,$ is any given smooth function, then the pair  $(\rho\,\mathbf{v},\bar{\lambda})$ is  a generalized  $\mathcal{C}^{\infty}-$symmetry of equation (\ref{Pn}) for  the function 
\[\bar{\lambda}={\lambda}-\frac{\mathbf{A}_n(\rho)}{\rho}. 
\]
In particular, when $\bar{\lambda}=0,$ i.e., if the function $\rho$ verifies $
\dfrac{\mathbf{A}_n(\rho)}{\rho}=\lambda,$ then $ \rho\mathbf{v}$  is a generalized symmetry.

Applying this result to the $\mathcal{C}^{\infty}-$symmetry  \eqref{lambdasympar} of equation (\ref{Pn}), we have that if a function $\rho=\rho(t,u^{(n-1)})$ verifies

\begin{equation}
\label{rho}
\dfrac{\mathbf{A}_n(\rho)}{\rho}=\lambda=\dfrac{u_1}{u}-k\,m\, u^m, 
\end{equation} then the generalized vector field  $\rho(x,u^{(n-1)})\,\partial_u$ becomes a generalized symmetry of the equation (\ref{Pn}).

In the next theorem we provide $n$ functions satisfying (\ref{rho}) which are used to construct $n$ generalized symmetries of  equation (\ref{Pn}):

\begin{thm}
	\label{teorema1}
	Let  $P_0:=u$ and  $P_n:=(\mathbf{D}_t+ku^m)(P_{n-1})$ for $n\geq 1.$ For $1\leq i\leq n,$
	the vector fields
	\begin{equation}
	\label{formula}
	\mathbf{w}_{n,i}:=u(P_{n-1})^{m-1}\left(\sum_{j=1}^{i}\dfrac{{(-1)^{j+1}}}{(i-j)!}t^{i-j}P_{n-j}\right)\partial_u,\qquad i=1,\dots,n
	\end{equation} define $n$ generalized symmetries of the $n$th-order equation $P_n=0.$
\end{thm}

\begin{proof}{\rm
		For $i=1,2,\cdots,n,$ let $\rho_{n,i}$ be the functions  defined by \begin{equation}\label{rhoi}
		\rho_{n,i}:=u(P_{n-1})^{m-1}H_i,  \quad \mbox{where} \quad H_i=\sum_{j=1}^{i}\dfrac{{(-1)^{j+1}}}{(i-j)!}t^{i-j}P_{n-j}.
		\end{equation} 
		By using the identities
		\begin{equation}\label{relaciones}
		\begin{array}{l}
		\mathbf{A}_n(P_{i-1})=P_i-ku^mP_{i-1} \quad \mbox{for} \quad i=1,2,\dots,n-1,\\
		\mathbf{A}_n(P_{n-1})=-ku^mP_{n-1},
		\end{array}
		\end{equation}
		which follow immediately from \eqref{cadenaparticular}, a straightforward calculation leads to
		\begin{equation}
		\label{salelambda}
		\begin{array}{lll}
		\dfrac{\mathbf{A}_n(\rho_{n,i})}{\rho_{n,i}}&=&\dfrac{\mathbf{A}_n(u)}{u}+\dfrac{\mathbf{A}_n\bigl((P_{n-1})^{m-1}\bigr)}{(P_{n-1})^{m-1}}+\dfrac{\mathbf{A}_n(H_i)}{H_i}\\[2 ex]
		&=& \dfrac{u_1}{u}-k(m-1)u^m-ku^m=\lambda.
		\end{array}
		\end{equation}
		This  shows that the functions given in (\ref{rhoi}) satisfy (\ref{rho}). The theorem follows immediately from the discussion included at the beginning of this section.
	}\end{proof}

\begin{rmk}\label{simrecursivas}
	{\rm It can be easily checked that the functions given in (\ref{rhoi})   satisfy $\partial_t(\rho_{n,i})=\rho_{n,i-1},$ for $1< i\leq n.$ Therefore,  the $n-1$ generalized   symmetries $\mathbf{w}_{n,i}=\rho_{n,i}\partial_u,$ for $1\leq i\leq n-1,$ can be directly determined from $ \mathbf{w}_{n,n}=\rho_{n,n}\partial_u$
		by using successive derivations with respect to $t,$ because $\partial_t^{i}(\rho_{n,n})=\rho_{n,n-i}.$   }
\end{rmk}
\section{First integrals derived from the generalized symmetries (\ref{formula})}\label{section3}

Let $\rho_1$ and $\rho_2$ be any functions satisfying (\ref{rho}). Then
\begin{equation}\label{cociente}
\dfrac{\mathbf{A}_n\left({\rho_1}/{\rho_2}\right)}{{\rho_1}/{\rho_2}}
=\dfrac{\mathbf{A}_n(\rho_1)}{\rho_1}-\dfrac{\mathbf{A}_n(\rho_2)}{\rho_2}=0. 
\end{equation} This implies that ${\rho_1}/{\rho_2}$ is a first integral for the $n$th-order equation (\ref{Pn}).
As a consequence of Theorem \ref{teorema1}, it follows that the functions
\begin{equation}
\label{ipAbel}
I_{(n;i,j)}:=\dfrac{\rho_{n,i}}{\rho_{n,j}}=\dfrac{H_i}{H_j},\quad i,j\in\{1,2,\dots,n\}
\end{equation}
are  first integrals of equation (\ref{Pn}).
In order to construct a complete set of first integrals, we can consider, for instance, the  $n-1$ first integrals:
\begin{equation}
\label{ip}
I_{(n;i)}:=\dfrac{\rho_{n,i}}{\rho_{n,1}}
=\displaystyle\sum_{r=1}^{i}\dfrac{{(-1)^{r+1}}}{(i-r)!}t^{i-r}Q_{n-r}  \quad (2\leq i\leq n),
\end{equation}
where $Q_{j}:={P_{j}}/{P_{n-1}},$  for $0\leq j\leq n-1.$

{\rm We observe that the functions (\ref{ip}) satisfy the relation
	$\partial_t(I_{(n;i)})=I_{(n;i-1)}, $ for $2< i\leq n;$ therefore
	the first integrals $I_{(n,2)},\ldots,I_{(n,n-1)}$ can be generated by successive derivations with respect to $t$  of  $I_{(n;n)}$ as follows:
	\begin{equation}
	\label{Jder}I_{(n;j)}=\partial_t^{n-j}(I_{(n;n)})\quad (2\leq j\leq n-1).\end{equation}    }

In the next theorem we prove that the $n-1$ first integrals
(\ref{ip}) are functionally independent:
\begin{thm}\label{teofirstintegrals}
	For $n\geq 2,$ the  functions defined by  (\ref{ip}) are $n-1$ functionally independent first integrals of the  $n$th-order equation (\ref{Pn}).
\end{thm}

\begin{proof}
	{\rm   Let   \begin{equation}\label{matrizM}
		{\cal{M}}=\dfrac{\partial(I_{(n;2)},\ldots, I_{(n;n)})}{\partial(t,u,u_1,\ldots,u_{n-1})}.
		\end{equation} denote  the  Jacobian matrix associated to the functions (\ref{ip}). Let ${\bar{\cal{M}}}$  be the square submatrix of ${\cal{M}}$   formed by its  last $n-1$ columns. Our goal is to check  that $\mbox{det}({\bar{\cal{M}}})\neq 0,$ in order to prove that the rank of  the Jacobian matrix ${\cal{M}}$ is $n-1.$
		
		The elements of the last row $R_{n-1}$ of  ${\bar{\cal{M}}}$ can be written in the form $\partial_{u_j}(I_{(n;n)})=a_{j;n-2}t^{n-2}+\cdots+a_{j;0},$ for $1\leq j\leq n-1,$    where
		\[	a_{j;n-l}=\dfrac{(-1)^{l-1}}{(n-l)!}\partial_{u_{j}}(Q_{n-l}),\quad \mbox{for}\quad 2\leq l\leq n.\]
		
		According to (\ref{Jder}), the elements of the  row $R_{n-2}$  are obtained by deriving the elements of $R_{n-1}$ with respect to $t$  and so on.  Consequently, for $1\leq i\leq n-1,$ the elements of the row $R_{i}$ of   ${\bar{\cal{M}}}$   are polynomials on $t$ of degree $i-1,$ and their respective constant terms, which will be used later, become
		\begin{equation}\label{terminosctes}
		a_{j;n-i-1}=(-1)^i\partial_{u_{j}}(Q_{n-i-1}),\quad \mbox{for}\quad 1\leq j\leq n-1.
		\end{equation}
		In order to cancel out the leading terms of the polynomials of the row $R_{i}$ of ${\bar{\cal{M}}}$  for $ 2\leq i\leq n-1,$ we replace  $R_i$ by
		\begin{equation}\label{rowi}
		R_{i}^{(1)}=R_{i}-\dfrac{t}{i-1}  R_{i-1}, 
		\end{equation} and denote   the resulting matrix by ${\bar{\cal{M}}}^{(1)}$. Clearly $\mbox{det}({\bar{\cal{M}}}^{(1)} )= \mbox{det}({\bar{\cal{M}}}),$ and  it is easy to check  that  all the elements of the new row $R_{i}^{(1)},$ for $i\geq 2,$ are polynomials on $t$ of degree $i-2$ and whose respective constant terms remain as in (\ref{terminosctes}).
		
		This process can be successively repeated until  getting  a matrix  ${\bar{\cal{M}}}^{(n-2)}$ whose elements do not depend on $t,$ and are given by (\ref{terminosctes}):
		\begin{equation}\label{matrizfinal}
		{\bar{\cal{M}}}^{(n-2)} =\Bigl((-1)^i\partial_{u_{j}}(Q_{n-i-1})
		\Bigr), \quad \mbox{for}\quad 1\leq i,j\leq n-1.
		\end{equation}
		
		Let $C_j$  denote the columns of  the matrix $\Bigl(\partial_{u_{j}}(Q_{n-i-1})
		\Bigr),$ for $1\leq i,j\leq n-1.$    By replacing the column $C_j$ by $C_j'$ where
		\begin{equation}\label{columnasnuevas}
		\begin{array}{l}
		C_j'=P_{n-1}\bigl(C_j-\partial_{u_{j}}(P_{n-1})C_{n-1}\bigr), \quad (1\leq j< n-1,n>2),\\C_{n-1}'=-P_{n-1}^2C_{n-1}
		\end{array}
		\end{equation} it can be checked 
		that matrix $\Bigl(\partial_{u_{j}}(Q_{n-i-1})
		\Bigr)$ is transformed into
		\begin{equation}\label{matrizfinalfinal}
		\left(
		\begin{array}{cccc}
		\partial_{u_1}(P_{n-2}) & \cdots &  1&  P_{n-2}\\
		\vdots& \ddots &\vdots &\vdots\\
		
		1
		& \cdots &  0&P_1\\
		0& \cdots &  0& P_0
		\end{array}
		\right).
		\end{equation} The  determinant of matrix (\ref{matrizfinalfinal}) is $P_0$ for  $n=2,3,$ and $(-1)^{k+1}P_0$ where $k>1$ is such that $n=2k$ or $n=2k+1.$
		By taking (\ref{matrizfinal}) and (\ref{columnasnuevas}) into account
		we finally deduce  that \begin{equation}
		\label{determinante}
		\delta_n:=\mbox{det}(\bar{\mathcal{M}})=\dfrac{P_0}{(P_{n-1})^n}\neq 0.\end{equation}  Therefore the rank of  matrix ${\mathcal{M}}$ is $n-1,$ which proves the theorem.}	
\end{proof}

\section{Complete integrability by using  Jacobi last multipliers}\label{section4}

In Theorem \ref{teofirstintegrals}, $n-1$ functionally independent first integrals of the $n$th-order equation $P_n=0$ have been obtained. It is clear that any other first integral of the form (\ref{ipAbel}) is functionally dependent of the mentioned  $n-1$ first integrals, because $I_{(n;k,l)}={I_{(n;k)}}/{I_{(n;l)}},$ for $1\leq l,k\leq n.$ In order to  integrate completely equation (\ref{Pn}), in this section we discuss how the Jacobi last multiplier (\ref{jlm}) can be used to find a remaining first integral.
With this objective we follow the next steps:
\begin{enumerate}
	\item As a consequence of Theorem  \ref{teofirstintegrals} the set
	\begin{equation}
	\Delta=\{(t,u^{(n-1)})\in M^{(n-1)}:I_{(n;j)}(t,u^{(n-1)})=c_{(n;j)},2\leq j\leq n\},
	\end{equation} where the $c_{(n;j)}$ are constants, is a submanifold of $M^{(n-1)}$ of dimension $2.$
	It follows from (\ref{determinante}) and the Implicit Function Theorem that $\Delta$ can be locally parametrized by $(t,u)$ as follows:
	$$\{(t,u^{(n-1)})\in M^{(n-1)}: u_j=G_{(n;j)}(t,u,c_{(n;2)},\dots,c_{(n;n)}), 2\leq j\leq n-2\},$$ for some smooth functions $G_{(n;j)}.$ The first order ODE \begin{equation}
	\label{auxiliar}
	u_1=G_{(n;1)}(t,u,c_{(n;2)},\dots,c_{(n;n)})
	\end{equation}  will be called the auxiliary equation. Next we deduce the explicit expression of this  auxiliary equation.
	
	The equations $I_{(n;i)}(t,u^{(n-1)})=c_{(n;i)},$ for $2\leq i\leq n-1,$ can be written in matrix form as follows
	\begin{equation}\label{sistema}
	\mathcal{N}_{n-1} \left(\begin{array}{l}
	P_{n-1}\\P_{n-2}\\\vdots\\P_2\\P_1
	\end{array}
	\right)=
	\left(\begin{array}{c}
	0\\0\\ \vdots\\0\\(-1)^{n}P_0
	\end{array}
	\right)
	\end{equation}
	where
	$$\mathcal{N}_{n-1}=\left(\begin{array}{cccc}
	t-c_{(n;2)}	& -1	    & \cdots  & 0 \\
	
	\dfrac{t^2}{2!}-c_{(n;3)}&  -t &\cdots   &  0\\
	
	\dfrac{t^3}{3!}-c_{(n;4)}&  -\dfrac{t^2}{2!} &\ \cdots  & 0 \\
	
	\vdots&  \vdots&   \ddots & \vdots \\
	\dfrac{t^{n-1}}{(n-1)!}-c_{(n;n)}& -\dfrac{t^{n-2}}{(n-2)!}&\cdots   & (-1)^nt \\
	\end{array} \right), (n\geq 2).$$
	It can be checked that   ${\mbox{det}(\mathcal{N}_{n-1})}$ is a  polynomial of $t$ of degree $n-1:$
	\begin{equation}\label{polinomios}
	S_{n-1}(t):={\mbox{det}(\mathcal{N}_{n-1})}=\sigma_n\Bigl(\dfrac{t^{n-1}}{(n-1)!}+\displaystyle\sum_{r=2}^{n}\dfrac{{(-1)^{r+1}}}{(n-r)!}t^{n-r}c_{(n;r)}\Bigr),
	\end{equation} where 
	\begin{equation}\label{c}\sigma_n=\left\{\begin{array}{lll}
	(-1)^{k+1} & \mbox{for} & n=2k,\\
	(-1)^{k} & \mbox{for} & n=2k+1.
	\end{array}\right.\end{equation}
	The following relation between the determinants of $\mathcal{N}_{n-1}$ and $\mathcal{N}_{n-2}$ can be easily proved:
	\begin{equation}\label{relaciondet}
	S_{n-1}'(t)=(-1)^nS_{n-2}(t).
	\end{equation}
	It follows from (\ref{sistema}) and (\ref{relaciondet})
	that
	$$\begin{array}{lll}
	P_1{\Bigl|_\Delta\Bigr.}&=&(-1)^{n}P_0\dfrac{{\mbox{det}(\mathcal{N}_{n-2})}}{{\mbox{det}(\mathcal{N}_{n-1})}}= P_0\,\dfrac{S_{n-1}'(t)}{S_{n-1}(t)}.\end{array}$$ Therefore the auxiliary equation (\ref{auxiliar}) becomes
	\begin{equation}\label{auxiliarfin}
	u_1=-ku^{m+1}+u \, \dfrac{S_{n-1}'(t)}{S_{n-1}(t)}.
	\end{equation}	
	This is a Bernoulli equation and it can be solved by a well-known procedure. Nevertheless, in what follows, we provide explicitly  its general solution by determining an integrating factor and an associated first integral.
	\item  According to   \cite[Corollary 6]{muriel2014lambdaJacobi}, an integrating factor of the auxiliary equation (\ref{auxiliar}) is the function
	$\mu=\mu(t,u,c_{(n;2)},\cdots,c_{(n;n)})$ defined by
	\begin{equation}
	\label{mu0}
	\mu:=\dfrac{(P_{n-1})^{-(n+m)}}{\delta_n}{\Bigl|_\Delta\Bigr.}=\dfrac{(P_{n-1})^{-m}}{ P_0}{\Bigl|_\Delta\Bigr.},
	\end{equation} where $\delta_n$ has been defined in (\ref{determinante}). The restriction of $P_{n-1}$ to $\Delta$ can be calculated from (\ref{sistema}) as
	$$ P_{n-1}{\Bigl|_\Delta\Bigr.}=\dfrac{P_0}	{\mbox{det}(\mathcal{N}_{n-1})}=\dfrac{P_{0}}{S_{n-1}(t)}.$$
	Therefore,  the integrating factor (\ref{mu0}) of equation (\ref{auxiliarfin}) becomes
	$${\mu}=u^{-(m+1)}\bigl(S_{n-1}(t)\bigr)^m.$$
	\item 	A corresponding  first integral $H=H(t,u,c_{(n;2)},\cdots,c_{(n;n)})$ of the exact first order equation ${\mu}(u_1-G_{(n;1)})=0$ can be calculated by a single quadrature and becomes:
	\begin{equation}\label{H}
	H:=u^{-m}\left(S_{n-1}(t)\right)^{m}-km\int\left(S_{n-1}(t)\right)^mdt.
	\end{equation}
	Consequently, the function $I_{(n;1)}=I_{(n;1)}(t,u^{(n-1)})$ defined by  \begin{equation}
	\label{ultima}
	I_{(n;1)}=H(t,u,I_{(n;2)},\dots, I_{(n;n)})
	\end{equation} is a first integral of the equation $P_n=0.$ Furthermore, the functions in $\{I_{(n;j)}\}_{j=1}^n$ are $n$ functionally independent first integrals of the equation $P_n=0.$	
	\item By setting $H=c_{(n;1)}$ in (\ref{H}), where $c_{(n;1)}\in \mathbb{R},$  the general solution of the equation $P_n=0$ can be expressed as
	\begin{equation}\label{chachi}
	\begin{array}{l}
	u(t)^m=\dfrac{\left(S_{n-1}(t)\right)^{m}}{km\displaystyle\int\left(S_{n-1}(t)\right)^mdt+c_{(n;1)}}.
	\end{array}\end{equation}
	
\end{enumerate}

Equation (\ref{chachi}) provides a formula to solve directly the $n$th-order equation $P_n=0$ by using the polynomials (\ref{polinomios}). These polynomials are the same for all the chains, because they are independent of the constants $k$ and $m.$ In order to simplify the expression of the polynomials (\ref{polinomios}) we set
$$\begin{array}{l}
C_{1}:=\sigma_n(n-1)!c_{(n;1)},\\
C_{r} :=(-1)^{r+1}\sigma_n\dfrac{(n-1)!}{(n-r)!}c_{(n;r)},\quad \mbox{for} \quad 2\leq r\leq n,
\end{array}$$
which permits to express the general solution of (\ref{Pn}) in terms of arbitrary polynomials of order $n-1,$ as it is stated in the next theorem for further reference:

\begin{thm}\label{teoremasol}
	Let  $P_0:=u$ and  $P_n:=(\mathbf{D}_t+ku^m)(P_{n-1})$ for $n\geq 1.$  	The general solution of the  equation $P_n=0$ is
	\begin{equation}
	\label{solucion}
	u(t)^m=\dfrac{\left(T_{n-1}(t)\right)^{m}}{km\displaystyle\int\left(T_{n-1}(t)\right)^mdt+C_{1}},
	\end{equation} where
	\begin{equation}\label{T}
	T_{n-1}(t)=t^{n-1}+C_2t^{n-1}+\dots+C_{n-1}t+C_{n},
	\end{equation} and  $C_i\in \mathbb{R}$ for $1\leq r\leq n.$
\end{thm}

In the following two sections we apply  these results to the Riccati and Abel chains.

\section{The Riccati chain}\label{seccionR}
The chain (\ref{cadenageneral}) generated by the function $g$  given by $g(u)=ku$, $u\in\mathbb{R}$, is known as the Riccati chain  of parameter $k\in \mathbb{R}$. We set 
\begin{equation}
\label{Ri}
R_0:=u\quad \mbox{and} \quad  R_i:=(\mathbf{D}_t+ku)(R_{i-1}),\quad i\geq 1.
\end{equation}
The equations in the Riccati chain are usually known as the higher-order Riccati equations. The  first four terms of this sequence define the ODEs displayed in Table \ref{tablagsR}:

\begin{figure}[htbp]
	\caption{Four first equations in the Riccati chain}\label{tablagsR}
	\centering  $ \begin{array}{c|l}
	\hline
	n & R_n=0   \\ \hline
	1    &  u_1+k u^2=0   \\
	2   &  u_2+k^2 u^3+3 k u u_1=0\\
	3   &u_3+ k \left(k^2 u^4+6 k u^2 u_1+4 u u_2+3 u_1^2\right)=0   \\
	4   & u_4+k \left(k^3 u^5+10 k^2 u^3 u_1+10 k u^2 u_2+5 u \left(3 k u_1^2+u_3\right)+10 u_1
	u_2\right)=0
	\end{array} $
\end{figure}

 \strut\vfill
%
%

\subsection{Generalized symmetries and  first integrals for the Riccati chain}

Theorem \ref{teorema1} can be used to determine $n$ generalized symmetries of the $n$th-order equation in the  considered chain. In this section we derive an additional generalized symmetry for the $n$th-order equation of the Riccati chain.
The term
\begin{equation}
\label{-1}
R_{-1}:=\dfrac{1}{k} 
\end{equation}  satisfies  \begin{equation}\label{relaciones0}
(\mathbf{D}_t+ku)(R_{-1})=R_0,
\end{equation}   which corresponds to (\ref{relaciones}) for $i=0.$ This suggests that the term (\ref{-1}) can be used   to  derive a new generalized symmetry for the Riccati chain:

\begin{thm}\label{teoremasimadicional}
	For $i\geq -1,$ let $R_i$ be the functions given in (\ref{Ri}) and  (\ref{-1}).
	The vector field
	\begin{equation}
	\label{nuevasR}
	\mathbf{w_R}_{n,n+1}:=u \left(\sum_{j=1}^{n+1}\dfrac{{(-1)^{j+1}}}{(n+1-j)!}t^{n+1-j}R_{n-j}\right)\partial_u
	\end{equation}
	is a generalized symmetry of the $n$th-order equation $R_n=0$ of the Riccati chain, for $n\geq 1.$
\end{thm}
\begin{proof}
	{\rm 	By  taking into account that $P_n=R_n$ for $m=1$, the vector field (\ref{nuevasR}) can be written in the form for   $\rho_{n,n+1}\partial_u$ where $\rho_{n,n+1}$ is the function defined in (\ref{rhoi}) but extended to $i=n+1:$
		\begin{equation}
		\label{nuevas}
		\begin{array}{l}	\rho_{n,n+1}  := u (P_{n-1})^{m-1}\displaystyle\sum_{j=1}^{n+1}\dfrac{{(-1)^{j+1}}}{(n+1-j)!}t^{n+1-j}P_{n-j} \quad (m=1).
		\end{array}
		\end{equation}
		By using (\ref{relaciones0}) and the corresponding identities (\ref{relaciones}):
		\begin{equation}\label{relacionesR}
		\begin{array}{l}
		\mathbf{A}_n(R_{i-1})=R_i-kuR_{i-1} \quad \mbox{for} \quad i=0,1,2,\dots,n-1,\\
		\mathbf{A}_n(R_{n-1})=-kuR_{n-1},
		\end{array}
		\end{equation}
		it can be  proved, as
		in the proof of Theorem \ref{teorema1}, that function (\ref{nuevas}) satisfies the corresponding condition (\ref{salelambda}) for $m=1:$
		\begin{equation}\label{salelambdaR}
		\dfrac{	\mathbf{A}_n(\rho_{n,n+1})}{\rho_{n,n+1}}=\dfrac{u_1}{u}-ku=\lambda.
		\end{equation}  The result follows immediately from the discussion at the beginning of Section \ref{section2}.
	} \end{proof}
The additional generalized symmetry (\ref{nuevasR})  for the $n$th-order equation of the Riccati chain permits to find, without any kind of additional integration, a new first integral, that together with the functionally independent first integrals derived in Theorem \ref{teofirstintegrals}, constitutes a complete set of first integrals for $R_n=0.$

\begin{coro}	\label{integralnuevariccati}
	The function
	\begin{equation}\label{InRiccati}
	I_{(n;n+1)}:=\dfrac{\rho_{n,n+1}}{\rho_{n,1}}=\dfrac{t^{n}}{n!}-\dfrac{t^{n-1}}{(n-1)!}\dfrac{R_{n-2}}{R_{n-1}}+\cdots+(-1)^{n+2}\dfrac{R_{-1}}{R_{n-1}}\end{equation}
	is a first integral of the $n$th-order equation $R_n=0$ of the Riccati chain. Moreover, (\ref{InRiccati}) and the functions given by (\ref{ip}) (for $m=1$) constitutes a complete set of first integrals of $R_n=0.$
\end{coro}
\begin{proof}    {\rm  According to (\ref{cociente}), the function (\ref{InRiccati}) is a first integral of $R_n=0$ because $\rho_{n,n+1}$ and $\rho_{n,1}$ satisfy (\ref{rho}) (see Equations (\ref{salelambda}) and (\ref{salelambdaR})).
		
		In order to prove the functional independence of the first integrals $I_{(n;j)},$ for $1\leq j\leq n+1,$ let   \begin{equation}\label{matrizMR}
		{\cal{M}}_R=\dfrac{\partial(I_{(n;2)},\cdots, I_{(n;n+1)})}{\partial(t,u,u_1,\cdots,u_{n-1})}
		\end{equation} denote  associated  Jacobian matrix. Let ${\bar{\cal{M}}_R}$  be the square submatrix of ${\cal{M}}_R$   formed by its  last $n$ columns.  As in the proof of Theorem \ref{teofirstintegrals}, it can be checked that $\mbox{det}({\bar{\cal{M}}_R})=\dfrac{R_{-1}}{(R_{n-1})^{n+1}}\neq 0,$ which proves  that the rank of  the Jacobian matrix ${\cal{M}}_R$ is $n.$ Therefore  $\{I_{(n;2)},\dots,I_{(n;n+1)}\}$ are $n$ functionally independent  first integrals of $R_n=0.$ }
	
	\end{proof}
\begin{rmk}
	{\rm It can be easily checked that the functions given in (\ref{rhoi}) and (\ref{nuevas}) satisfy $\partial_t(\rho_{n,i})=\rho_{n,i-1}$ for $1\leq i\leq n+1.$ Therefore  the $n$ generalized symmetries of $R_n=0$ derived in Theorem \ref{teorema1} (for $m=1$) can be directly determined from (\ref{nuevas})
		by using successive derivations with respect to $t.$  For instance, the corresponding generalized symmetries (\ref{formula}) and (\ref{nuevasR})
		for the fourth-order equation in the Riccati chain become
		\begin{equation}\label{wsRiccati} \begin{array}{l}
		\mathbf{w_R}_{(4,5)}=   u\,\left(R_3\dfrac{t^4}{4!}-R_2\dfrac{t^3}{3!}+R_1\dfrac{t^2}{2}-R_0t+R_{-1}\right)\partial_u,   \\   [2ex] \mathbf{w_R}_{(4,4)}= u\,\left(R_3\dfrac{t^3}{3!}-R_2\dfrac{t^2}{2}+R_1t-R_0\right)\partial_u    ,
		\\   [2ex]\mathbf{w_R}_{(4,3)}=    u\,\left(R_3\dfrac{t^2}{2}-R_2t+R_1\right)\partial_u,\\   [2ex]
		\mathbf{w_R}_{(4,2)}=u\,\left(R_3t-R_2\right)\partial_u,\\   [2ex]
		\mathbf{w_R}_{(4,1)}=u\,R_3 \partial_u.   \end{array}  \end{equation}

		Similarly, it follows from (\ref{Jder}) that a complete set of first integrals of $R_n=0$ can be easily computed from the first integral
		(\ref{InRiccati})
		by successive derivations with respect to $t.$  For instance, the corresponding first integrals (\ref{ip}) and (\ref{InRiccati})
		for the fourth-order equation in the Riccati chain become
		\begin{equation}\label{IsRiccatiorden4} \begin{array}{l}
		I_{(4,5)}=\dfrac{t^4}{4!}-\dfrac{R_2}{R_3} \dfrac{t^3}{3!}  +\dfrac{R_1}{R_3}\dfrac{t^2}{2}  -\dfrac{R_0}{R_3}t+\dfrac{R_{-1}}{R_3} ,   \\   [2ex] I_{(4,4)}=\dfrac{t^3}{3!}-\dfrac{R_2}{R_3} \dfrac{t^2}{2}  +\dfrac{R_1}{R_3}t-\dfrac{R_0}{R_3} ,   \\   [2ex] I_{(4,3)}=\dfrac{t^2}{2}-\dfrac{R_2}{R_3}{t}+ \dfrac{R_1}{R_3},
		\\   [2ex]	I_{(4,2)}=          t-\dfrac{R_2}{R_3} .   \end{array}  \end{equation}
}\end{rmk}

\subsection{General solutions to the equations of the Riccati chain}

For the Riccati chain (i.e., for $m=1$), the corresponding formula (\ref{solucion})  becomes
\begin{equation}\label{chachiR}
u(t)=\dfrac{T_{n-1}(t)}{k\displaystyle\int T_{n-1}(t)dt+C_1},
\end{equation}  which provides directly the general solutions of $R_n=0,$ by using the  polynomials given in (\ref{T}). In the next table we give the general solutions for the four first equations in the Riccati chain (see Table \ref{tablagsR}), and where the $C_i$ are arbitrary real constants:

\begin{figure}[htbp]
	\label{SR} 	\caption{General solutions for the four first equations  in the Riccati chain} \bigskip
	\centering  $
	\def\arraystretch{2}
	\begin{array}{l|l}
	\hline	n     & \mbox{General Solution of} \quad R_n=0 \\ \hline
	{1}   &      u \left( t \right) ={\dfrac {1}{k{t}+C_{
				{1}}}}
	\\
	{2}   &      u \left( t \right) ={\dfrac {2\,(t+C_{2})}{k\left(t^{2}+2\,C_{{2}}\,t\right)+2\,C_{
				{1}}}}
	\\
	{3}   &    u \left( t \right) ={\dfrac {3!\,({t}^{2}+C_{{2}}t+C_{{3}})}{k\left({t}^{3}+3
			\,C_{{2}}{t}^{2}+6\,C_{{3}}t\right)+6\,C_{{1}}}}
	\\
	{4}   &   u \left( t \right) ={\dfrac {4!\,({t}^{3}+C_{{2}}{t}^{2}+C_{{3}}t+C
			_{{4}})}{k\left({t}^{4}+4\,C_{{2}}{t}^{3}+6\,C_{{3}}{t}^{2}+12\,C_{{4}}t\right)+12
			\,C_{{1}}}}
	
	
	\\
	\end{array} $
\end{figure}

\section{The Abel chain}\label{seccionA}

The chain (\ref{cadenageneral}) generated by the function $g$  given by $g(u)=ku^2$, $u\in\mathbb{R}$, is known as the Abel chain  of parameter $k\in \mathbb{R}$. We set \begin{equation}
\label{Ai}
A_0:=u\quad \mbox{and} \quad  A_i:=(\mathbf{D}_t+ku^2)(A_{i-1}),\quad i\geq 1.
\end{equation}
The first three terms of the Abel chain are the following ODEs:

\begin{figure}[htbp]
	\caption{Three first equations in the Abel chain}\label{tablagsA}
	\centering  $ \begin{array}{c|l}
	\hline
	n & A_n=0   \\ \hline
	1    &  u_1+k u^2=0   \\
	2   &  u_2+k^2 u^5+4 k u_1 u^2=0\\
	3   &u_3+ k (k^2 u^4+6 k u^2 u_1+4 u u_2+3 u_1^2)=0   \\
	\end{array} $
\end{figure}

 \strut\vfill
%
%

For the Abel chain, the parameter $m$ is $m=2$ and the corresponding formula (\ref{solucion})  becomes
\begin{equation}\label{chachiA}
u(t)^2=\dfrac{(T_{n-1}(t))^2}{k\displaystyle\int (T_{n-1}(t))^2dt+C_1}.
\end{equation}   In  Table \ref{SA} we give the general solutions for the three first equations in the Abel chain shown in Table \ref{tablagsA}:
\begin{figure}[htbp]
	\caption{  General solutions for the Abel chain}\label{SA}\bigskip
	\centering  $ 	\def\arraystretch{2}
	\begin{array}{l|l}
	n     & \mbox{General Solution of} \quad A_n=0 \\ \hline
	{1}   &  u(t)^2=\dfrac{1}{2kt+C_1}
	
	\\
	
	{2}   &  u(t)^2=\dfrac{3(t+C_2)^2}{2k(t+C_2)^3+3C_1}
	
	\\
	
	{3}   & u(t)^{2}=\dfrac{15(t^2+C_2t+C_3)^2}{k\left(6t^5+15C_2t^4+10(C_2^2+2C_3)t^3+(30C_3C_2)t^3+30C_3^2t\right)+15C_1}

	\end{array} $
\end{figure}

These solutions should be compared with those obtained in \cite[Theorem 3.2]{muriel2014lambda}, which were obtained from the solutions of the higher Riccati equations by solving an additional first-order ODE (see Eq. (41) in \cite{muriel2014lambda}).  The  procedure introduced in this paper provides directly the general solutions by means of (\ref{chachiA}), which does not require either the solutions of the Riccati higher equations or any type of additional integration.

\section{Example for $m \not \in\mathbb{Z}$}\label{section7}
The results obtained in Sections \ref{seccion2}-\ref{section4} have been established for sequences of equations of the form (\ref{cadenageneral}) generated by $g(u)=ku^m,$ for arbitrary values $k \in\mathbb{R}$ and $m\in\mathbb{Z}.$ In this section we present an example  in order to show these results are also valid for real values of parameter $m,$  with adequate restrictions on the involved domains.

The third-order ODE
\begin{equation}\label{ejemplo}
\sqrt{u}u_3+14u_2u+5{u_1^2}+72 u_1u\sqrt{u}+64 {u^3}=0,\quad u>0
\end{equation} is the element $P_3=0,$ of the sequence (\ref{cadenageneral}) generated by the function $g(u)=4 \sqrt{u},$ which correspond to $k=4,m=1/2.$ Therefore, Theorem \ref{teoremasol} can be applied to provide directly its general solution through (\ref{solucion}):
\begin{equation}\label{solej}
u(t)=\dfrac{T_2(t)}{\left((C_3-C_2^2)\ln\Bigl|C_2+t+\sqrt{|T_2(t)|}\Bigr|+(t+C_2)\sqrt{|T_2(t)|}+C_1\right)^2},
\end{equation}
where $$T_2(t)={t^2+2C_2t+C_3},\quad C_i\in\mathbb{R}\quad \mbox{for} \quad  i=1,2,3.$$
Although the Lie symmetry group of equation (\ref{ejemplo}) is two-dimensional (and therefore  insufficient to complete its integration by quadratures), Theorem \ref{teorema1} provides three generalized symmetries of the equation,
\begin{equation}\label{simeje} \begin{array}{l}
\mathbf{w}_{(3,3)}= u\,(P_2)^{-1/2}\left(P_2\dfrac{t^2}{2}-P_1t+P_0\right)\partial_u    ,
\\   [2ex]
\mathbf{w}_{(3,2)}=    u\,(P_2)^{-1/2}\left(P_2t-P_1\right)\partial_u,\\   [2ex]
\mathbf{w}_{(3,1)}=u\,(P_2)^{1/2} \partial_u,   \end{array}  \end{equation}	
by using the preceding elements in the sequence:
$$P_0=u,\quad P_1=u_1+4u\sqrt{u},\quad  P_2=u_2+10u_1\sqrt{u}+16u^2.$$
By Theorem \ref{teofirstintegrals}, these terms can also be used to obtain, without any integration, the first integrals
\begin{equation}\label{Isejemplo} \begin{array}{l}
I_{(3,3)}=\dfrac{t^2}{2}-\dfrac{P_1}{P_2}t-\dfrac{P_0}{P_2} ,   \quad 	I_{(3,2)}=          t-\dfrac{P_1}{P_2} .   \end{array}  \end{equation}  By using  (\ref{polinomios})
$$S_2(t)=\mbox{det}(\mathcal{N}_{2})=\left|\begin{array}{cc}
t-c_{(3;2)}	& -1	    \\

\dfrac{t^2}{2!}-c_{(3;3)} & -t \\
\end{array} \right|=-\dfrac{t^2}{2}+tc_{(3,2)}-c_{(3,3)}$$
a remaining first integral can be calculated from (\ref{ultima}), which can be expressed as follows
$$  I_{(3,1)}=(\sqrt{u} -t+I_{(3,2)}))J+\sqrt{2}\left(I_{(3,3)}-\dfrac{(I_{(3,2)})^2}{2}\right)\arctan\left(\dfrac{t-I_{(3,2)}}{J}\right),$$ where  $J=J(t,u^{(2)})$ is the function given by $$J=\sqrt{\left|-\dfrac{t^2}{2}+tI_{(3,2)}-I_{(3,3)}\right|}.$$
\section{Concluding remarks and extensions}

In this paper, we have determined a unified method to study a family of differential sequences, in order to obtain their first integrals, generalized symmetries and exact solutions.   For any $n$th-order equation in each chain, we have obtained a set of $n$ generalized symmetries in evolutionary form,   
and verified the essential property that the ratio of any  two characteristics become a first integral of the equation (Theorem \ref{teorema1}). Furthermore, we have demonstrated taht $n-1$ of these first integrals are functionally independent (Theorem \ref{teofirstintegrals}). It is noteworthy that
the generalized symmetries can be easily derived from one another by means of simple derivations with respect to $t$ (see  Remark \ref{simrecursivas} and equation (\ref{Jder})).

In order to obtain a complete set of first integrals, we have exploited the knowledge of a Jacobi last multiplier for each equation of the considered sequences. Thus, a remaining first integral is determined by quadrature only. Finally, we have shown that the complete set of first integrals yields the general solution of each $n$th-order equation of any sequence and can be expressed through (\ref{solucion}) in terms of arbitrary polynomials of order $n-1$ (Theorem \ref{teoremasol}).

We have applied our general results to the Riccati and Abel chains. In particular,   we have derived an additional generalized symmetry in the case of the Riccati chain. Consequently, the $n+1$ determined generalized symmetries yield a complete set of first integrals, without any kind of integration and the use of the Jacobi last multiplier.

The existence of an additional generalized symmetry with similar properties for the other sequences of the family, e.g. the Abel chain, remains as an a open problem. Another issue which needs to be investigated further is how the generalized symmetries obtained in this work can be used to determine the  complete symmetry groups of each equation, not only in the Riccati chain (which have been derived in \cite{andriopoulosdifferential} in terms of nonlocal symmetries, and in \cite{nucci2016} as Lie point symmetries of the equivalent first-order equations), but for any sequence in the family and, in particular, in the Abel chain.

%
%

\subsection*{Acknowledgements}

C. Muriel acknowledges  the financial support of the Junta de Andaluc\'ia
research group FQM-377 and from FEDER/Ministerio de Ciencia, Innovaci\'on y Universidades-Agencia Estatal de Investigaci\'on/Proyecto PGC2018-101514-B-I00.

\label{lastpage}
\end{document}